\documentclass{vldb}
\usepackage{amsmath, alltt, amssymb, xspace, times, epsfig, comment, url, graphicx}
\usepackage{amssymb, comment, algorithm, algpseudocode}
\usepackage{verbatim}
\usepackage{multirow}
\usepackage{balance}
\usepackage{color}

\newtheorem{proposition}{Proposition}

\newtheorem{definition}{Definition}

\newcommand\discuss[1]{\{\textbf{Discuss:} \textit{#1}\}}
\newcommand\todo[1]{\{\textbf{Todo:} \textit{#1}\}}

\newcommand\note[1]{\{\underline{Note:} \textit{#1}\}}

\newcommand{\mypara}[1]{\vspace*{0.06in}\noindent\textbf{#1}\xspace}

\newcommand{\difp}{differential privacy\xspace}
\newcommand{\km}{$k$-means\xspace}

\newcommand{\Lap}[1]{\ensuremath{\mathsf{Lap}\left(#1\right)}\xspace}
\newcommand{\Var}[1]{\ensuremath{\mathsf{Var}\left(#1\right)}\xspace}

\newcommand{\E}{\ensuremath{\mathsf{E}}\xspace}
\newcommand{\EE}[1]{\ensuremath{\mathsf{E}\left[#1\right]}\xspace}
\newcommand{\mse}[1]{\ensuremath{\mathsf{MSE}\left(#1\right)}\xspace}

\newcommand{\BI}[1]{\ensuremath{\mathsf{Bias}\left(#1\right)}}

\renewcommand{\Pr}[1]{\ensuremath{\mathsf{Pr}\left[#1\right]}\xspace}
\newcommand{\myexp}[1]{\ensuremath{e^{#1}}\xspace}
\newcommand{\MM}{\ensuremath{\mathcal{M}}\xspace}

\renewcommand{\AA}{\ensuremath{\mathcal{A}}\xspace}

\newcommand{\cj}{\ensuremath{\mathbf{o}^{j}}\xspace}
\newcommand{\CJ}{\ensuremath{C^{j}}\xspace}

\newcommand{\nicvminus}{NICV$^{-}$}

\newcommand{\ug}{\ensuremath{\mathsf{UG}}\xspace}
\newcommand{\eug}{\ensuremath{\mathsf{EUG}}\xspace}
\newcommand{\ag}{\ensuremath{\mathsf{AG}}\xspace}

\newcommand{\dpl}{DPLloyd\xspace}

\newcommand{\dplscopt}{DPLloyd-OPT\xspace}
\newcommand{\gupt}{GkM\xspace}
\newcommand{\privGene}{PGkM\xspace}
\newcommand{\gene}{Gene\xspace}
\newcommand{\augkm}{AUGkM\xspace}
\newcommand{\eugkm}{EUGkM\xspace}
\newcommand{\ugkm}{UGkM\xspace}

\newcommand{\mkm}{MkM\xspace}

\begin{document}
\title{Differentially Private $k$-Means Clustering}


\author{%
	{Dong Su$^{~\#}$, Jianneng Cao$^{~*}$, Ninghui Li$^{~\#}$, Elisa Bertino$^{~\#}$, Hongxia Jin$^{~\dagger}$}%
	\vspace{1.6mm}\\
	\fontsize{10}{10}\selectfont\itshape
	$^{\#}$\,Department of Computer Science, Purdue University\\
	\fontsize{9}{9}\selectfont\ttfamily\upshape
	\{su17, ninghui, bertino\}@cs.purdue.edu%
    \vspace{1.2mm}\\
    \fontsize{10}{10}\selectfont\rmfamily\itshape
    $^{*}$\,Institute for Infocomm Research, Singapore\\
    \fontsize{9}{9}\selectfont\ttfamily\upshape
    \,caojn@i2r.a-star.edu.sg
    \vspace{1.2mm}\\
    \fontsize{10}{10}\selectfont\rmfamily\itshape
    $^{\dagger}$\, Samsung Information Systems of America \\
    \fontsize{9}{9}\selectfont\ttfamily\upshape
    \,hongxia.jin@sisa.samsung.com
}

\maketitle

\begin{abstract}
There are two broad approaches for differentially private data analysis.  The \emph{interactive} approach aims at developing customized differentially private algorithms for various data mining tasks.
The \emph{non-interactive approach} aims at developing differentially private algorithms that can output a synopsis of the input dataset, which can then be used to support various data mining tasks.
In this paper we study the tradeoff of interactive vs.~non-interactive approaches and propose a \emph{hybrid} approach that combines interactive and non-interactive, using $k$-means clustering as an example.  In the hybrid approach to differentially private \km clustering, one first uses a non-interactive mechanism to publish a synopsis of the input dataset, then applies the standard \km clustering algorithm to learn $k$ cluster centroids, and finally uses an interactive approach to further improve these cluster centroids.  We analyze the error behavior of both non-interactive and interactive approaches and use such analysis to decide how to allocate privacy budget between the non-interactive step and the interactive step.
Results from extensive experiments support our analysis and demonstrate the effectiveness of our approach.
\end{abstract}

\sloppypar

\graphicspath{{figures/PINQ/}{figures/orig/}{figures/gowalla/}{figures/gowalla_3d/}{figures/s1/}{figures/GUPT/}{figures/hybrid/}{figures/3D/}{figures/equation-11-15/}{figures/varyN/}{figures/kmeans/}{figures/scalability/}{figures/find_opt_m/}{figures/vldb/}{figures/performance/}{figures/vldb/dplloyd}{figures/vldb}{figures/performance/}{figures/vldb/aug}}

\section{Introduction}\label{sec:intro}


In recent years, a large and growing body of literature has investigated differentially private data analysis. Broadly, they can be classified into two approaches. The \emph{interactive} approach aims at developing customized differentially private algorithms for specific data mining tasks.  One identifies the queries that need to be answered for the data mining task, analyze their sensitivity, and then answers them by adding appropriate noises.
The \emph{non-interactive approach} aims at developing an approach to compute, in a differentially private way, a synopsis of the input dataset, which can then be used to generate a synthetic dataset, or to directly support various data mining tasks.

An intriguing question is which of the two approaches is better?  Given an input dataset $D$, the desired privacy parameter $\epsilon$, which we refer to as the privacy budget, and one or more data analysis tasks, should one use the interactive approach or the non-interactive approach?  This question is largely open.
In general, the non-interactive approach has the advantage that once a synopsis is constructed, many analysis tasks can be conducted on the synopsis.  In contrast, using the interactive approach, one is limited to executing the interactive algorithm just once; any additional access to the dataset would violate differential privacy.  Therefore, strictly speaking, a dataset can serve only one analyst, and for only one task. (One could divide the privacy budget for multiple analysts and/or multiple tasks, but then the accuracy for each task will suffer.)  On the other hand, because the interactive approach is designed specifically for a particular data mining task, one might expect that, under the same privacy budget it should be able to produce more accurate results than the non-interactive approach.


In this paper we initiate the study of the tradeoff of interactive vs.~non-interactive approaches, using $k$-means clustering as the example.  Clustering analysis plays an essential role in data management tasks.  Clustering has also been used as a prime example to illustrate the effectiveness of interactive differentially private data analysis~\cite{BDMN05,Dwo11,PINQ,McSherry09,MTS+12,NRS07,ZXY+13}.
There are three state of the art interactive algorithms.  The first is the differentially private version of the Lloyd algorithm~\cite{BDMN05, McSherry09}, which we call \mbox{\dpl}.  The second algorithm uses the sample and aggregation framework~\cite{NRS07} and is implemented in the GUPT system~\cite{MTS+12}, which we call \gupt.  The third and most recent one, which we call \privGene, uses PrivGene~\cite{ZXY+13}, a framework for differentially private model fitting based on genetic algorithms.

To the best of our knowledge, performing $k$-means clustering using the non-interactive approach has not been explicitly proposed in the literature.  In this paper, we propose to combine the following non-interactive differentially private synopsis algorithms with \km clustering.  The dataset is viewed as a set of points over a $d$-dimensional domain, which is divided into $M$ equal-size cells, and a noisy count is obtained from each cell.  A key decision is to choose the parameter $M$.  A larger $M$ value means lower average counts for each cell, and therefore noisy counts are more likely to be dominated by noises.  A smaller $M$ value means larger cells, and therefore one has less accurate information of where the points are.  We propose a method that sets $M=\left(\frac{N\epsilon}{10}\right)^{\frac{2d}{2+d}}$, which is derived based on extending the analysis in~\cite{QYL12}, which aims to minimize errors when answering rectangular range queries for $2$-dimensional data, to higher dimensional case.  We call the resulting $k$-means algorithm EUGkM, where EUG is for Extended Uniform Grid.

We conducted extensive experimental evaluations for these algorithms on 6 external datasets and 81 datasets that we synthesized by varying the dimension $d$ from $2$ to $10$ and the number of clusters from $2$ to $10$.
Experimental results are quite interesting. \gupt was introduced after \dpl and was claimed to have accuracy advantage over \dpl, and \privGene was introduced after and compared \gupt.  However, we found that \dpl is the best method among these three methods.  In the comparison of \dpl and \gupt in~\cite{MTS+12}, \dpl was run using much larger number of iterations than necessary, and thus perform poorly.  In~\cite{ZXY+13}, \privGene was compared only with \gupt, and not with \dpl.  
More specifically, we found that \gupt is by far the worst among all methods.  Through experimental analysis of the sources of the errors, we found that it is possible to dramatically improve the accuracy of \gupt by choosing smaller partitions in the sample and aggregation framework.   After this improvement, \gupt becomes competitive with \privGene.  However, \dpl, the earliest method is clearly the best performing algorithm among the $3$ interactive algorithms.  Through analysis, we found that why \dpl outperforms \privGene.  
The genetic programming style \privGene needs more iterations to converge.  When making these algorithms differentially private, the privacy budget is divided among all iterations, thus having more iterations means more noise is added to each iteration.  Therefore, the more direct \dpl outperforms \privGene.



The most intriguing results are those comparing \dpl with EUGkM. For most datasets, EUGkM performs much better than \dpl.  For a few, they perform similarly, and for two datasets \dpl outperforms EUGkM.  Through further theoretical and empirical analysis, we found that while the performance of both algorithms are greatly affected by the two key parameters $d$ and $k$, they are affected differently by these two parameters.  \dpl scales worse when $k$ increases, while EUGkM scales worse when $d$ increases.  Again we use analysis to demonstrate why this is the case.


An intriguing question is can we further improve \dpl?  The accuracy of \dpl is affected by two key factors: the number of iterations and the choice of initial centroids.  In fact, these two are closely related.  If the initially chosen centroids are very good and close to the true centroids, one only needs perhaps one iteration to improve it, and this reduction in the number of iterations would mean little noise is added.  

This leads us to propose a novel hybrid method that combines non-interactive EUGkM with interactive \dpl.  We first use half the privacy budget to run EUGkM, and then use the centroids outputted by EUGkM as the initial centroids for one round of \dpl.  Such a method, however, may not actually outperform EUGkM, especially when the privacy budget $\epsilon$ is small, since then one round of \dpl may actually worsen the centroids.
We use our error analysis formulas to determine whether there is sufficient privacy budget for such a hybrid approach to outperform EUGKM.  We then experimentally validate the effectiveness of the Hybrid approach.

The hybrid idea is applicable to general private data analysis tasks which require parameter tuning.  In the no-privacy setting, one typically tunes parameters by building models for several parameters and selecting the one which offers the best utility.  Under the differential privacy setting, such kind of parameter tuning procedure does not work well since the limited privacy budget might be over-divided by trying many different parameters.  Chaudhuri et al.~\cite{CMS11} proposed a method for private parameter tuning by taking advantage of parallel composition.  The idea is to build private models with different parameters on separate subset of the dataset and evaluate models on a validation set.  The best parameter is chosen via exponential mechanism with quality function defined by the evaluation score.  However, this approach is also not scalable well over a large set of candidate parameters which might result each data block to have very small number of points and therefore lead to very inaccurate model.  Our proposed hybrid approach offers a better solution.  We can first publish private synopses of the input data, on which we try a large set of parameters.  Then, we run the interactive private analysis with the selected parameter on the input dataset to get the final result.

In this paper we advance the state of art on differentially private data mining in several ways.  First,
we have introduced non-interactive methods for differentially private $k$-means clustering, which are highly effective and often outperform state of the art interactive methods.  Second, we have extensively evaluated three interactive methods, and one non-interactive methods, and analyzed their strengths and weaknesses.  Third, we have developed techniques to analyze the error resulted from both \dpl and EUGkM.
Finally, we introduce the novel concept of hybrid approach to differentially private data analysis, which is so far the best approach to $k$-means clustering.  We conjecture that the concept of hybrid differential privacy approach may prove useful in other analysis tasks as well.

The rest of the paper is organized as follows.  In Section~\ref{sec:related}, we discuss related work.  In Section~\ref{sec:preliminaries}, we give preliminary information about differential privacy and \km clustering.  In Section~\ref{sec:existingApproaches}, we describe the existing three interactive approaches, \dpl, \gupt, \privGene and 
one non-interactive approache \eugkm.  
In Section~\ref{sec:towardHybrid}, we first show the experimental results on the performance comparison among the interactive and non-interactive approaches, and analyze their strengths and weaknesses.  In Section~\ref{sec:hybrid} we study the error behavior of \dpl and \eugkm, introduce the hybrid approach, and compare these with existing algorithms.  We conclude in Section~\ref{sec:conclusions}.

\section{Related Work}\label{sec:related}
The notion of differential privacy was developed in a series of papers \cite{DN03,DN04,BDMN05,DMNS06,Dwo06}.  Several primitives for answering a single query differentially privately have been proposed.  Dwork et al.~\cite{DMNS06} introduced the method of adding Laplacian noise scaled with the sensitivity.  McSherry and Talwar~\cite{MT07} introduced a more general exponential mechanism.  Nissim et al. \cite{NRS07} proposed adding noises proportion to local sensitivity.


Blum et al.~\cite{BDMN05} proposed a sublinear query (SuLQ) database model for interactively answering a sublinear number (in the size of the underlying database) of count queries differential privately. The users (e.g. machine learning algorithms) issue queries and get responses which are added laplace noises. They applied the SuLQ framework to the \km clustering and some other machine learning algorithms.
McSherry \cite{McSherry09} built the PINQ (Privacy INtegrated Queries) system, a programming platform which provides several differentially-private primitives to enable data analysts to write privacy-preserving applications. These private primitives include noisy count, noisy sum, noisy average, and exponential mechanism. The \dpl algorithm, which we compare against in this paper, has been implemented using these primitives.
Another programming framework with differential privacy support is Airavat, which makes programs using the MapReduce framework differentially private~\cite{RSK+10}.

Nissim et al. \cite{NRS07,Smith11} propose the sample and aggregate framework (SAF), and use $k$-means clustering as a motivating application for SAF.
This SAF framework has been implemented in the GUPT system~\cite{MTS+12} and is evaluated by \km clustering.  This is the \gupt algorithm that we compared with in the paper.
Dwork~\cite{Dwo11} suggested applying a geometric decreasing privacy budget allocation strategy among the iterations of \km, whereas we use an increasing sequence.  Geometric decreasing sequence will cause later rounds using increasingly less privacy budget, resulting in higher and higher distortion with each new iteration. Zhang et al. ~\cite{ZXY+13} proposed a general private model fitting framework based on genetic algorithms. The \privGene approach in this paper is an instantiation of the framework to $k$-means clustering.

Interactive methods for other data mining tasks have been proposed.
McSherry and Mironov~\cite{MM09} adapted algorithms producing recommendations from collective user behavior to satisfy differential privacy.  Friedman and Schuster~\cite{FS10} made the ID3 decision tree construction algorithm differentially private.  Chaudhuri and Monteleoni~\cite{CM08} proposed a differentially private logistic regression algorithm.  Zhang et al.~\cite{ZZX+12} introduced the functional mechanism, which perturbs an optimization objective to satisfy differential privacy, and applied it to linear regression and logistic regression.  Differentially private frequent itemset mining has been studied in~\cite{BLST10,LQSC12}.  The tradeoffs of interactive and non-interactive approaches in these domains are interesting future research topics.

Most non-interactive approaches aim at developing solutions to answer histogram or range queries accurately~\cite{DMNS06,XWG11,HRMS10,CPS+12}. Dwork et al.~\cite{DMNS06} calculate the frequency of values and release their distribution differentially privately. Such method makes the variance of query result increase linearly with the query size. To address this issue, Xiao et al. \cite{XWG11} propose a wavelet-based method, by which the variance is polylogarithmic to the query size. Hay et al. \cite{HRMS10} organize the count queries in a hierarchy, and improve the accuracy by enforcing the consistency between the noisy count value of a parent node and those of its children. Cormode et al.~\cite{CPS+12} adapted standard spatial indexing techniques, such as quadtree and kd-tree, to decompose data space differential-privately.  Qardaji et al.~\cite{QYL12} proposed the UG and AG method for publishing 2-dimensional datasets.  Mohammed et al. \cite{MCFY11} tailored the non-interactive data release for construction of decision trees. 

Roth et al.~\cite{BLR08} studied the problem on how to release synthetic data differentially privately for any set of count queries specified in advance.  They proposed a $\epsilon$-differentially private mechanism whose error scales only logarithmically with the number of queries being answered.  However, it is not computationally efficient (super-polynomial in the data universe size).   
Subsequent work includes~\cite{DNR09, HR10, DRV10, RR10, GHRU11, HLM12}.  One of the typical works is the private multiplicative weight mechanism~\cite{HR10} which is proposed to answer count queries interactively whose error also scales logarithmically with the number of queries seen so far.  Its running time is only linear in the data universe size.

\section{Background}\label{sec:preliminaries}

\subsection{Differential Privacy} \label{sec:preliminaries:differential}

Informally, differential privacy requires that the output of a data analysis mechanism should be approximately the same, even if any single tuple in the input database is arbitrarily added or removed.
\begin{definition}[{$\epsilon$-Differential Privacy~\cite{Dwo06,DMNS06}}] \label{def:diff}
A randomized mechanism $\AA$ gives $\epsilon$-differential privacy if for any pair of neighboring datasets $D$ and $D'$, and any $S\in \mathit{Range}(\AA)$,
$$\Pr{\AA(D)=S} \leq e^{\epsilon}\cdot \Pr{\AA(D')=S}.$$
\end{definition}

In this paper we consider two datasets $D$ and $D'$ to be neighbors if and only if either $D=D' + t$ or $D'=D + t$, where $D +t$ denotes the dataset resulted from adding the tuple $t$ to the dataset $D$. We use $D\simeq D'$ to denote this.  This protects the privacy of any single tuple, because adding or removing any single tuple results in $e^{\epsilon}$-multiplicative-bounded changes in the probability distribution of the output.


Differential privacy is composable in the sense that combining multiple mechanisms that satisfy differential privacy for $\epsilon_1, \cdots,\epsilon_m$ results in a mechanism that satisfies $\epsilon$-differential privacy for $\epsilon=\sum_{i} \epsilon_i$.  Because of this, we refer to $\epsilon$ as the privacy budget of a privacy-preserving data analysis task.  When a task involves multiple steps, each step uses a portion of $\epsilon$ so that the sum of these portions is no more than $\epsilon$.


There are several approaches for designing mechanisms that satisfy $\epsilon$-differential privacy, including Laplace mechanism \cite{DMNS06} and Exponential mechanism \cite{MT07}. The Laplace mechanism computes a function $g$ on the dataset $D$ by adding to $g(D)$ a random noise, the magnitude of which depends on $\mathsf{GS}_g$, the \emph{global sensitivity} or the $L_1$ sensitivity of $g$.  Such a mechanism $\AA_g$ is given below:
$
\begin{array}{crl}
& \AA_g(D) & =g(D)+\mathsf{Lap}\left(\frac{\mathsf{GS}_g}{\epsilon}\right)
  \\
\mbox{where} &  \mathsf{GS}_g & = \max\limits_{(D,D') : D \simeq D'} |g(D) - g(D')|,
\\
\mbox{and}& \Pr{\Lap{\beta}=x} & = \frac{1}{2\beta} \myexp{-|x|/\beta}.
\end{array}
$

In the above, $\Lap{\beta}$ denotes a random variable sampled from the Laplace distribution with scale parameter $\beta$. 

\subsection{$\mathbf{k}$-means Clustering Algorithms}\label{sec:preliminaries:kmeans}

The $k$-means clustering problem is as follows: given a $d$-dimensional dataset $D=\{x^1, x^2,\dots, x^N\}$, partition data points in $D$ into $k$ sets $\mathbf{O}=\{O^1, O^2, \cdots, O^k\}$ so that the
Normalized Intra-Cluster Variance (NICV) is minimized
\begin{equation}\label{eqn:NICV}
\frac{1}{N} {\sum_{j=1}^{k} \sum_{x^\ell \in O^j} ||x^\ell - o^j ||^2}.
\end{equation}

The standard \km algorithm is the Lloyd's algorithm \cite{Lloyd82}. The algorithm starts by selecting $k$ points as the initial choices for the centroid.
The algorithm then tries to improve these centroid choices iteratively until no improvement can be made. In each iteration, one first partitions the data points into $k$ clusters, with each point assigned to be in the same cluster as the nearest centroid.
Then, one updates each centroid to be the center of 
the data points in the cluster.
\begin{equation}
\forall i\in [1..d]\; o^{j}_i \leftarrow \frac{\sum_{x^\ell \in O^{j}} x^\ell_i}{|O^{j}|}, \label{eq:hardK-update}
\end{equation}
where $j=1,2,\dots, k$, $x^\ell_i$ and $o^j_i$ are the $i$-th dimensions of $x^\ell$ and $o^j$, respectively.
The algorithm continues by alternating between data partition and centroid update, until it converges.

\section{The Interactive and Non-interactive Approaches}\label{sec:existingApproaches}

In this section, we describe 3 interactive approaches and 2 non-interactive approaches to differential private $k$-means clustering.

\subsection{Interactive Approaches}

\subsubsection{\dpl}\label{sec:ping}
Differentially private \km or LLoyd's algorithm was first proposed by Blum et al.~\cite{BDMN05} and was later implemented in the PINQ system~\cite{McSherry09}, a platform for interactive privacy preserving data analysis.  We call this the \dpl approach.
\dpl differs from the standard Lloyd algorithm in the following ways.  First, Laplacian noise is added to the iterative update step in the Lloyd algorithm.  Second, the number of iterations needs to be fixed in order to decide how much noise needs to be added in each iteration.

Each iteration requires computing the total number of points in a cluster and, for each dimension, the sum of the coordinates of the data points in a cluster.  Let $t$ be the number of iterations, and $d$ be the number of dimensions.
Then, each tuple is involved in answering $dt$ sum queries and $t$ count queries. To bound the sensitivity of the sum query to a small number $r$, each dimension is normalized to $[-r, r]$. Thus, the global sensitivity of \dpl is $(dr+1)t$, and each query is answered by adding Laplacian noise $\Lap{\frac{(dr+1)t}{\epsilon}}$.

There are two issues that greatly impact the accuracy of \dpl.  The first is the number of iterations.  A large number of iterations causes too much noises being added.  A small number of iterations may be insufficient for the algorithm to converge.  In~\cite{McSherry09}, the number of iterations is set to be $5$, which seems to work well for many settings.
The second is the quality of initial centroids.  A poor choice of initial centroids can result in converging to a local optimum that is far from global optimum, or not converging after the given number of iterations.  While many methods for choosing the initial points have been developed~\cite{PLL99}, these methods were developed without the privacy concern and need access to the dataset.  In~\cite{McSherry09}, $k$ points at uniform random from the domain are chosen as the initial centroids.  We have observed empirically that this can perform poorly in some settings, since some randomly chosen initial centroids are close together.  We thus introduce an improved method for choosing initial centroids that is similar to the concept of sphere packing.   Given a radius $a$, we randomly generate $k$ centroids one by one such that each new centroid is of distance at least $a$ away from each border of the domain and each new centroid is of distance at least $2a$ away from any existing centroid.  When a randomly chosen point does not satisfy this condition, we generate another point.  When we have failed repeatedly, we conclude that the radius $a$ is too large, and try a smaller radius.  We use a binary search to find the maximal value for $a$ such that it is the process of choosing $k$ centroids succeed. This process is data independent.

\subsubsection{\gupt}\label{sec:gupt}

The \km clustering problem was also used to motivate the {\em sample and aggregate} framework (SAF) for satisfying \difp, which was developed in \cite{NRS07,Smith11}, and implemented in the GUPT system~\cite{MTS+12}.  



Given a dataset $D$ and a function $f$, SAF first partitions $D$ into $\ell$ blocks, then it evaluates $f$ on each of the block, and finally it privately aggregates results from all blocks into a single one.  
Since any single tuple in $D$ falls in one and only one block, adding one tuple can affect at most one block's result, limiting the sensitivity of the aggregation step.  Thus one can add less noise in the final step to satisfy differential privacy.
%

As far as we know, GUPT \cite{MTS+12} is the only implementation of SAF.  Authors of~\cite{MTS+12} implemented \km clustering and used it to illustrate the effectiveness of GUPT.  We call this algorithm \mbox{\gupt}.
Given a dataset $D$, it first partitions $D$ into $\ell$ \mbox{blocks} $D_1, D_2, \dots, D_\ell$.
Then, for each block $D_b$ ($1\leq b\leq \ell$), it calculates its $k$ centroids $o^{b,1}, o^{b,2}, \dots, o^{b,k}$.
Finally, it averages the centroids calculated from all blocks and adds noise. Specifically, the $i$'th dimension of the $j$'th aggregated centroid is
\begin{equation}\label{eqn:GKM_SAF}
o^j_i = \frac{1}{\ell}\sum_{b=1}^\ell o^{b,j}_i +\Lap{\frac{2(max_i-min_i)\cdot k\cdot d}{\ell\cdot \epsilon}},
\end{equation}
where $o^{b,j}_i$ is the $i$'th dimension of $o^{b,j}$, $[min_i, max_i]$ is the estimated output range of $i$'th dimension.
One half of the total privacy budget is used to estimate this output range, and the other half is used for adding Laplace noise.

We have found that the implementation downloaded from \cite{GUPT}, which uses Equation (\ref{eqn:GKM_SAF}), performed poorly.  Analyzing the data closely, we found that $min_i$ and $max_i$ often fall outside of the data range, especially for small $\epsilon$. We slightly modified the code to bound $min_i$ and $max_i$ to be within the data domain.  This does not affect the privacy, was able to greatly improve the accuracy.  In this paper we use this fixed version.

Here a key parameter is the choice of $\ell$.  Intuitively, a larger $\ell$ will result in each block being very small and unable to preserve the cluster information in the blocks, and a smaller $\ell$, on the other hand, results in large noise added.  (Note the inverse dependency on $\ell$ in Equation (\ref{eqn:GKM_SAF}).  Analysis in~\cite{MTS+12} suggests to set $\ell=N^{0.4}$.  Our experimental results, however, show that the performance is quite poor.  We consider a variant that chooses $\ell = \frac{N}{3k}$, i.e., having each block containing $3k$ points, which performs much better than setting $\ell=N^{0.4}$.

\subsubsection{\privGene}\label{sec:privGene}

PrivGene~\cite{ZXY+13} is a general-purpose differentially private model fitting framework based on genetic algorithms. Given a dataset $D$ and a fitting-score function $f(D, \theta)$ that measures how well the parameter $\theta$ fits the dataset $D$, the PrivGene algorithm initializes a candidate set of possible parameters $\theta$ and iteratively refines them by mimicking the process of natural evolution.  Specifically, in each iteration, PrivGene uses the exponential mechanism \cite{MT07} to privately select from the candidate set $m^{\prime}$ parameters that have the best fitting scores, and generates a new candidate set from the $m^{\prime}$ selected parameters by crossover and mutation. Crossover regards each parameter as an $\ell$-dimensional vector. Given two parameter vectors, it randomly selects a number $\bar{\ell}$ such that $0<\bar{\ell}<\ell$ and splits each vector into the first $\bar{\ell}$ dimensions in the vector and the remaining $\ell-\bar{\ell}$ dimensions (the lower half). Then, it swaps the lower halves of the two vectors to generate two child vectors. These vectors are then mutated by adding a random noise to one randomly chosen dimension.

In \cite{ZXY+13}, PrivGene is applied to logistic regression, SVM, and \km clustering. In the case of \km clustering, the NICV formula in Equation \ref{eqn:NICV}, more precisely its non-normalized version, is used as the fitting function $f$, and the set of $k$ cluster centroids is defined as parameter $\theta$. Each parameter is a vector of $\ell=k\cdot d$ dimensions.  
Initially, the candidate set is populated with $200$ sets of cluster centroids randomly sampled from the data space, each set containing exactly $k$ centroids. Then, the algorithm runs iteratively for $\max\{8, (xN\epsilon)/m^{\prime}\}$ rounds, where $x$ and $m^{\prime}$ are empirically set to $1.25\times 10^{-3}$ and $10$, respectively, and $N$ is the dataset size.



We call the approach of applying PrivGene to \km clustering \privGene, which is similarly to \dpl in that it tries to iteratively improve the centroids.  However, rather than maintaining and improving a single set of $k$ centroids, \privGene maintains a pool of candidates, uses selection to improve their quality, and crossover and mutation to broaden the pool.  Similar to \dpl, a key parameter is the number of iterations.  Too few iterations, the algorithm may not converge.  Too many iterations means too little privacy budget for each iteration, and the exponential mechanism may not be able to select good candidates.

\subsection{Non-interactive Approaches}

Interactive approaches such as \dpl and \gupt suffer from two limitations.  First, often times the purpose of conducting \km clustering is to visualize how the data points are partitioned into clusters.  The interactive approaches, however, output only the centroids.  In the case of \dpl, one could also obtain the number of data points in each cluster; however, it cannot provide more detailed information on what shapes data points in the clusters take.  The value of interactive private \km clustering is thus  limited.  Second, as the privacy budget is consumed by the interactive method, one cannot perform any other analysis on the dataset; doing so will violate differential privacy.


Non-interactive approaches, which first generate a synopsis of a dataset using a differentially private algorithm, and then apply \km clustering algorithm on the synopsis, avoid these two limitations.  In this paper, we consider the following synopsis method.  Given a $d$-dimensional dataset, one partitions the domain into $M$ equal-width grid cells, and then releases the noisy count in each cell, by adding Laplacian noise to each cell count.

The synopsis released is a set of cells, each of which has a rectangular bounding box and a (noisy) count of how many data points are in the bounding box.  The synopsis tells only how many points are in a cell, but not the exact locations of these points.  For the purpose of clustering, We treat all points as if they are at the center of the bounding box.  In addition, these noisy counts might be negative, non-integer, or both.  A straightforward solution is to round the noisy count of a cell to be a non-negative nearest integer and replicate the cell center as many as the rounded count. This approach, however, may introduce a significant systematic bias in the clustering result, when many cells in the \ug synopsis are empty or close to empty and these cells are not distributed uniformly.  In this case, simply turning negative counts to zero can produce a large number of points in those empty areas, which can pull the centroid away from its true position.  We take the approach of keeping the noisy count unchanged and adapting the centroid update procedure in \km to use the cell as a whole.  Specifically, given a cell with center $c$ and noisy count $\tilde{n}$, its contribution to the centroid is $c\times \tilde{n}$.  Using this approach, in one cluster, cells who have negative noisy count can ``cancel out'' the effect of other cells with positive noise.  Therefore, we can have better clustering performance.

For this method, the key parameter is $M$, the number of cells.  When $M$ is large, the average count per cell is low, and the noise will have more impact.  When $M$ is small, each cell covers a large area, and treating all points as at the center may be inaccurate when the points are not uniformly distributed.  We now describe two methods of choosing $M$.

\subsubsection{\eugkm}\label{sec:eug}

Qardaji et al.~\cite{QYL12} studied the effectiveness of producing differentially private synopses of 2-dimensional datasets for answering rectangular range counting queries (i.e., how many data points there are in a rectangular range) with high accuracy, and suggested choosing $M=\frac{N\epsilon}{10}$. 
We now analyze the choice of $M$ for higher-dimensional case.  We use {\em mean squared error} (MSE) to measure the accuracy of $est$ with respect to $act$. That is,
$$
\mse{est}=\EE{(est-act)^2}=\Var{est}+(\BI{est})^2,
$$
where $\Var{est}$ is the variance of $est$ and $\BI{est}$ is its bias.

There are two error sources when computing $est$. First, Laplace noises are added to cell counts to satisfy \difp. This results in the variance of $est$. Since counting a cell size has the sensitivity of 1, Laplace noise $\Lap{\frac{1}{\epsilon}}$ is added. Thus, the noisy count has the variance of $\frac{2}{\epsilon^2}$. Suppose that the given counting query covers $\alpha$ portion of the total $M$ cells in the data space. Then, $\Var{est}=\alpha\frac{2M}{\epsilon^2}$. Second, the given counting query may not fully contain the cells that fall on the border of the query rectangle. To estimate the number of points in the intersection between the query rectangle and the border cells, it assumes that data are uniformly distributed. This results in the bias of $est$, which depends on the number of tuples in the border cells. The border of the given query consists of $2d$ hyper rectangles, each being $(d-1)$-dimensional.
The number of cells falling on a hyper rectangle is in the order of $M^{\frac{d-1}{d}}$. On average the number of tuples in these cells is in the order of $M^{\frac{d-1}{d}}\cdot \frac{N}{M}=\frac{N}{M^{\frac{1}{d}}}$. Therefore, we estimate the bias of $est$ with respect to one hyper rectangle to be $\beta\frac{N}{M^{\frac{1}{d}}}$, where $\beta\geq 0$ is a parameter. We thus estimate $(\BI{est})^2$ to be $2d\left(\beta\frac{N}{M^{\frac{1}{d}}}\right)^2$. Summing the variance and the squared bias, it follows that
$$
\mse{est}= \alpha\frac{2M}{\epsilon^2} + \beta^2\frac{2dN^2}{M^{\frac{2}{d}}}.
$$
To minimize the MSE, we set the derivative of the above equation with respect to $M$ to 0. This gives
\begin{equation}\label{eqn:eug_formula}
M=\left(\frac{N\epsilon}{\theta}\right)^{\frac{2d}{2+d}},
\end{equation}
where $\theta = \sqrt{\frac{\alpha}{2\beta^2}}$.
We name the above extended approach as \eug (extended uniform griding approach). We use \eugkm to represent the \eug-based \km clustering scheme.

\section{Performance and Analysis}\label{sec:towardHybrid}

In this section, we compare and analyze the performance of the five methods introduced in the last section.


\subsection{Evaluation Methodology}\label{sec:existingAppExpt}

We experimented with six external datasets and a group of syntheticly generated datasets.  The first dataset is a 2D synthetic dataset S1~\cite{SIPU}, which is a benchmark to study the performance of clustering schemes.  S1 contains 5,000 tuples and 15 Gaussian clusters.  The Gowalla dataset contains the user checkin locations from the Gowalla location-based social network whose users share their checking-in time and locations (longitude and latitude).  We take all the unique locations, and obtain a 2D dataset of 107,091 tuples.  We set $k = 5$ for this dataset. The third dataset is a 1-percentage sample of road dataset which was drawn from the 2006 TIGER (Topologically Integrated Geographic Encoding and Referencing) dataset \cite{TIGER}.  It contains the GPS coordinates of road intersections in the states of Washington and New Mexico.   The fourth is Image~\cite{SIPU}, a 3D dataset with 34,112 RGB vectors. We set $k = 3$ for it.  We also use the well known Adult dataset \cite{AN10}. We use its six numerical attributes, and set $k=5$. The last dataset is Lifesci. It contains 26,733 records and each of them consists of the top 10 principal components for a chemistry or biology experiment. As previous approaches \cite{MTS+12,ZXY+13}, we set $k=3$. Table \ref{table:dis} summarizes the datasets. For all the datasets, we normalize the domain of each attribute to [-1.0, 1.0].

When generating the synthetic datasets, we fix the dataset size to 10,000, and vary $k$ and $d$ from 2 to 10. For each dataset, $k$ well separated Gaussian clusters of equal size are generated, and 30 sets of initial centroids are generated in the same way as in Section \ref{sec:ping}.



Implementations for \dpl and GkM were downloaded from \cite{PINQ} and \cite{GUPT}, respectively. The source code of \privGene~\cite{ZXY+13} was shared by the authors. 
We implemented \eugkm.


\begin{table}[h]
\caption{Descriptions of the Datasets.}
\label{table:dis}
\scriptsize
\begin{center}
\begin{tabular}{ | c | c | c | c | c | c |}
\hline
Dataset & \# of tuples & $d$ & $k$ & $\ell_{\mbox{GkM}}$ & $\ell_{\mbox{GkM-3K}}$ \\
\hline
S1 & 5,000 & 2 & 15 & 30 & 111\\
\hline
Gowalla & 107,091 & 2  & 5 & 103 & 7,139\\
\hline
TIGER & 16,281& 2  & 2 & 48 & 2,714\\
\hline
Image & 34,112 & 3 & 3 & 65 & 3,790\\
\hline
Adult-num & 48,841 & 6 & 5 & 75 & 3,256\\
\hline
Lifesci & 26,733 & 10 & 3 & 59 & 2,970\\
\hline
Synthetic & 10,000 & [2, 10] & [2, 10] & 40 & $10000/(3k)$\\
\hline

\end{tabular}
\end{center}
\end{table}

\mypara{Configuration.} Each algorithm outputs $k$ centroids $\mathbf{o}=\{o^1,o^2,\cdots,o^k\}$.  To evaluate the quality of such an output $\mathbf{o}$, we compute the average squared distance between any data point in $D$ and the nearest centroid in $\mathbf{o}$, and call this the NICV.




We note that since both \dpl and \eugkm use Lloyd-style iteration,  
they are affected by the choice of initial centroids. In addition, all algorithms have random noises added somewhere to satisfy differential privacy. To conduct a fair comparison, we need to carefully average out such randomness effects.
\gupt and \privGene do not take a set of initial centroids as input.  \gupt divides the input dataset into multiple blocks, and for each block invokes the standard \km implementation from the Scipy package~\cite{Scipy} with a different set of initial centroids to get the result, and finally aggregates the outputs for all the blocks.  We run \gupt and \privGene 100 times and report the average result.

For \dpl, we generate 30 sets of initial centroids, run \dpl 100 times on each set of initial centroids, and we report the average of the 3000 NICV values as the final evaluation of \dpl.  

The non-interactive approach (\eugkm) has the advantage that once a synopsis is published, one can run \km clustering with as many sets of initial centroids as one wants and choose the result that has the best performance relative to the synopsis.
In our experiments, given a synopsis, we use the same $30$ sets of initial centroids as those generated for the \dpl method. For each set, we run clustering and output a set of $k$ centroids. Among all the 30 sets of output centroids, we select the one that has the lowest NICV relative to the synopsis rather than to the original dataset. This process ensures selecting the set of output centroids satisfies differential privacy. We then compute the NICV of this selected set relative to the original dataset, and take it as the resulting NICV with respect to the synopsis.  To deal with the randomness introduced by the process of generating synopsis, we generate 10 different synopses and take the average of the resulting NICV.

As the baseline, we run standard \km algorithm \cite{Lloyd82} over the same 30 sets of initial centroids and take the minimum NICV among all the 30 runs.

\mypara{Experimental Results.}
Figure~\ref{fig:non-interactive_VS_interactive} reports the results for the 6 external datasets.  For these, we vary $\epsilon$ from $0.05$ to $2.0$ and plot the NICV curve for the methods mentioned in Section~\ref{sec:existingApproaches}.  This enables us to see how these algorithms perform under different $\epsilon$.

Figure~\ref{fig:synthe-heatmap} reports the results for the synthetic datasets.  For these, we fix $\epsilon = 1.0$ and report the difference of NICV between each approach and the baseline.  This enables us to see the scalability of these algorithms when $k$ and $d$ increase.

For interactive approaches, \dpl has the best performance in most cases. Its performance is worse than that of \privGene only on the small dataset S1 when the privacy budget $\epsilon$ is smaller than $0.15$. 
Comparing \dpl and \eugkm, we observe that
in the four low dimensional datasets (S1, Gowalla, TIGER and Image), \eugkm clearly outperforms \dpl at small $\epsilon$ value and their gap becomes smaller as $\epsilon$ increases. However, in the two high dimensional datasets (Adult-num and Lifesci), \dpl outperforms \eugkm almost in all given privacy budgets.
Similar results can also be found in Figure \ref{fig:synthe-heatmap}.  Figure~\ref{fig:synthe-heatmap} also exhibits the effects of the number of clusters and the number of dimensions.  The \eugkm's performance is more sensitive to the increase of dimension, while \dpl gets worse quickly as the number of clusters increases.  Below we analyze these algorithms to understand why they perform in this way.  In addition, Figure~\ref{fig:synthe-heatmap} shows the difference of \eugkm's performance on different $\theta$ choices.  Setting $\theta = 10$ for \eugkm works well in most cases.

\begin{figure*}[!htb]

	\begin{tabular}{cc}
    \multicolumn{2}{c}{\hspace{-2mm}\includegraphics[width = \textwidth]{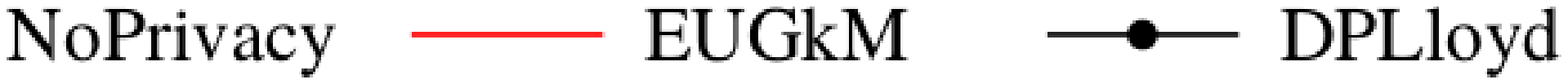}}\\    
	\includegraphics[width = 3.2in]{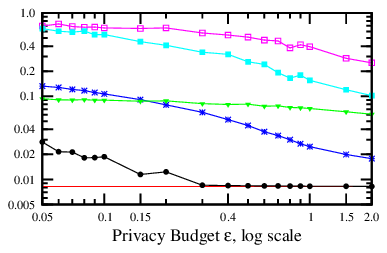} &
	\includegraphics[width = 3.2in]{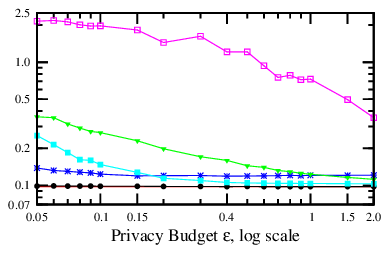}\\
	(a) S1 [$d = 2$, $k = 15$] & (b) Image [$d = 3$, $k = 3$]
	\end{tabular}

	\begin{tabular}{cc}
	\includegraphics[width = 3.2in]{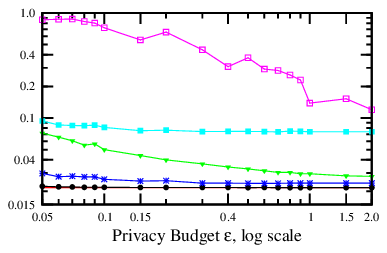} &
	\includegraphics[width = 3.2in]{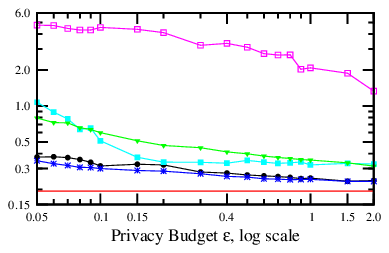}\\
	(c) Gowalla [$d = 2$, $k = 5$]& (d) Adult-num [$d = 6$, $k = 5$]
	\end{tabular}

	\begin{tabular}{cc}
	\includegraphics[width = 3.2in]{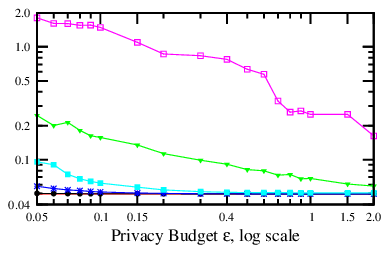} &
	\includegraphics[width = 3.2in]{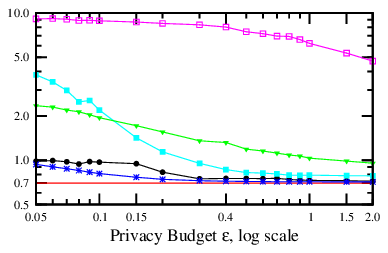}\\
	(e) TIGER [$d = 2$, $k = 2$] & (f) Lifesci [$d = 10$, $k = 3$]
	\end{tabular}

	\caption{The comparison of \dpl, \eugkm, \privGene and \gupt. x-axis: privacy budget $\epsilon$ in log-scale. y-axis: NICV in log-scale.}\label{fig:non-interactive_VS_interactive}
\end{figure*}

\begin{figure*}[!htb]

\begin{tabular}{ccc}
	\hspace{-0.5cm}\includegraphics[width = 2.3in]{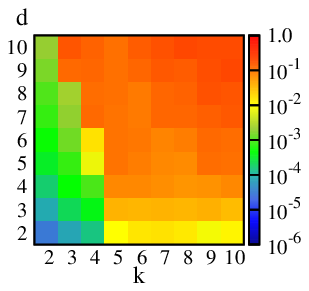} 
	& \hspace{-0.5cm} \includegraphics[width = 2.3in]{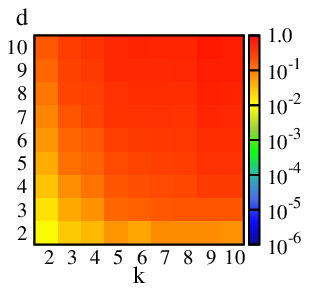}
	& \hspace{-0.5cm} \includegraphics[width = 2.3in]{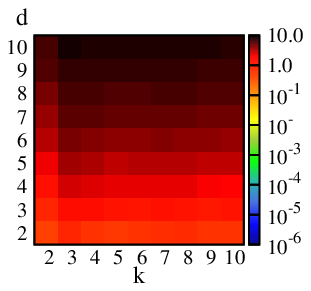} \\
	\hspace{-0.5cm}(a) \dpl & \hspace{-0.5cm} (b) \privGene  & \hspace{-0.5cm}(c) \gupt  \\
	
	\hspace{-0.5cm}	\includegraphics[width = 2.3in]{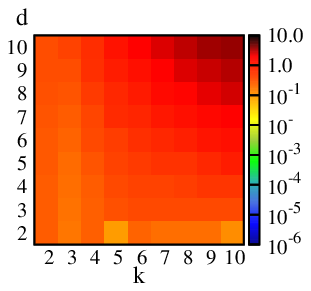}
	& \hspace{-0.5cm}	\includegraphics[width = 2.3in]{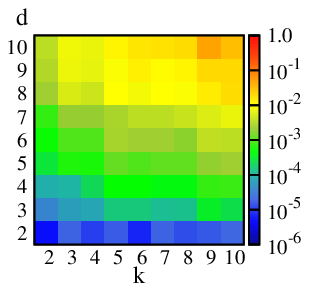}
	&\hspace{-0.5cm}	\includegraphics[width = 2.3in]{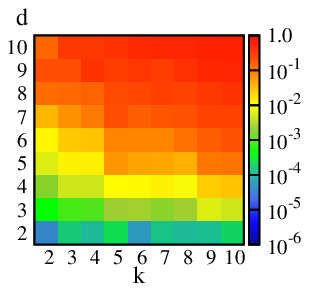}\\
	\hspace{-0.5cm}(d) \gupt-3K & \hspace{-0.5cm} (e) \eugkm & \hspace{-0.5cm} (f) \eugkm $\theta = 1$ \\

	\hspace{-0.5cm}	\includegraphics[width = 2.3in]{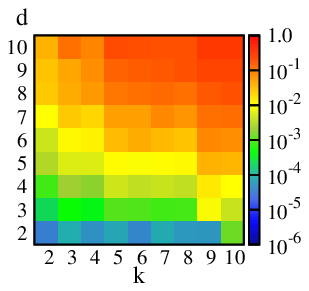}
	& \hspace{-0.5cm}	\includegraphics[width = 2.3in]{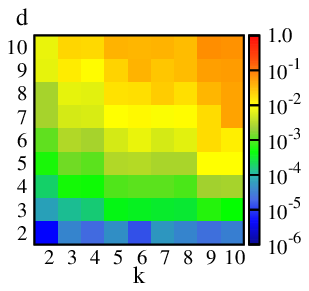}
	& \hspace{-0.5cm}	\includegraphics[width = 2.3in]{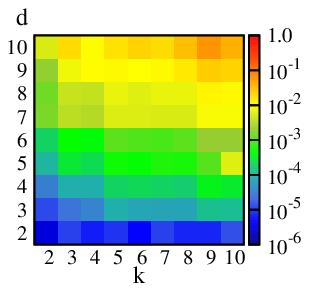}\\
	\hspace{-0.5cm}(g) \eugkm $\theta = 2$ & \hspace{-0.5cm} (h) \eugkm $\theta = 5$ & \hspace{-0.5cm} (i) \eugkm $\theta = 20$ \\

	\hspace{-0.5cm}	\includegraphics[width = 2.3in]{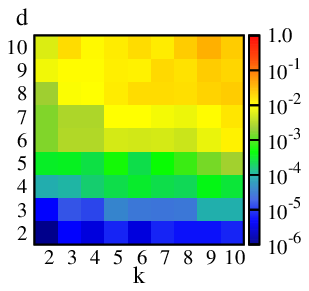}
	& \hspace{-0.5cm}	\includegraphics[width = 2.3in]{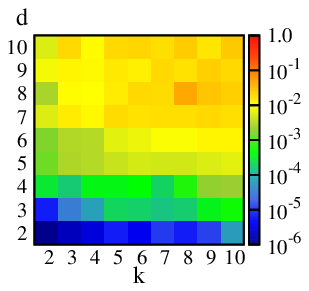}
	& \quad\\
	\hspace{-0.5cm}(j) \eugkm $\theta = 50$ & \hspace{-0.5cm} (k) \eugkm $\theta = 100$  & \quad \\

\end{tabular}

	\caption{The heatmap by varying $k$ and $d$}\label{fig:synthe-heatmap}
\end{figure*}

\subsection{The Analysis of the \gupt Approach}\label{sec:non_inter:gkm}

From Figures~\ref{fig:non-interactive_VS_interactive} and \ref{fig:synthe-heatmap}, it is clear that \gupt is always much worse than others.  There are two sources of errors for \gupt.  One is that \gupt is aggregating centroids computed from the subsets of data, and this aggregate may be inaccurate even without adding noise.  The other is that the noise added according to Equation~(\ref{eqn:GKM_SAF}) may be too large.  To tease out the role played by these two error sources, Figure \ref{fig:gupt-only} shows the effect of varying block size from around $N^{0.1}$ to $N$.  It shows error from \gupt, error from using the aggregation without noise (SAG), and error from adding noise computed by Equation \ref{eqn:GKM_SAF}) to the best known centroids (Noise).  From the figure, it is clear that setting $\ell=N^{0.4}$, which corresponds to block size of $N^{0.6}$ is far from optimal, as the error \gupt is dominated by that from the noise, and is much higher than the error due to sample and aggregation.
Indeed, we observed that as the block size decreases the error of \gupt keeps decreasing, until when the block size gets close to $k$.  It seems that even though many individual blocks result in poor centroids, aggregating these relatively poor centroids can result in highly accurate centroids.  This effect is most pronounced in the Tiger dataset, which consists of two large clusters.  The two centroids computed from each small block can be approximately viewed as choosing one random point from each cluster.  When averaging these centroids, one gets very close to the true centroids.

This observation motivated the introduction of \gupt-3K algorithm, which fixes each block size to be $3k$.  Recall that we are to select $k$ centroids from each block.  As can be seen from Figures~\ref{fig:non-interactive_VS_interactive} and \ref{fig:synthe-heatmap}, \gupt-3K becomes competitive with PGkM, sometimes significantly outperforms \privGene (e.g. TIGER and Lifesci), although it still underperforms \dpl.



\begin{figure*}[!htb]
	\begin{tabular}{ccc}
	\includegraphics[width = 2.2in]{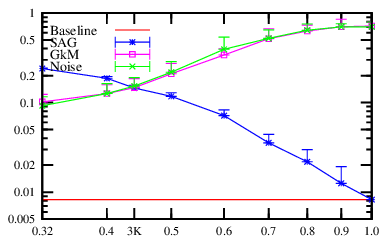} &
	\includegraphics[width = 2.2in]{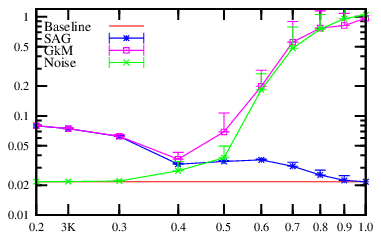} &
	\includegraphics[width = 2.2in]{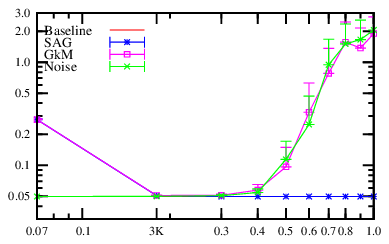} \\
	(a) S1 [$d = 2$, $k = 15$]  & (b) Gowalla [$d = 2$, $k = 5$] & (c) TIGER [$d = 2$, $k = 2$]
	\end{tabular}

	\begin{tabular}{ccc}
	\includegraphics[width = 2.2in]{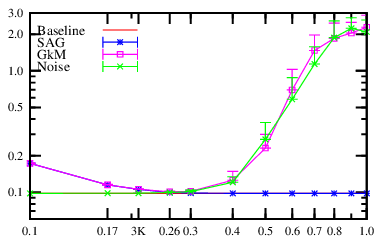} &
	\includegraphics[width = 2.2in]{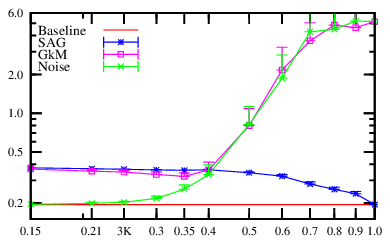} &
	\includegraphics[width = 2.2in]{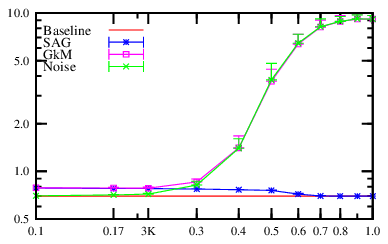}\\
	(d) Image [$d = 3$, $k = 3$]  & (e) Adult-num [$d = 6$, $k = 5$]  & (f) Lifesci [$d = 10$, $k = 3$]
	\end{tabular}

    \caption{The analysis of the \gupt Approach. x-axis: block size exponent in log-scale, y-axis: NICV in log-scale. }\label{fig:gupt-only}
\end{figure*}





\subsection{The Analysis of the \privGene Approach}

\privGene is a stochastic \km method based on genetic algorithms. A stochastic method converges to global optimum \cite{KrishnaM99}. On the contrary, \dpl is a gradient descent method derived from the standard Lloyd's algorithm \cite{Lloyd82}, which may reach local optimum. However, \privGene is still inferior to \dpl in Figure~\ref{fig:non-interactive_VS_interactive}.

There are two possible reasons.
First, a stochastic approach typically takes a `larger' number of iterations to converge \cite{KrishnaM99}. Figure \ref{fig:convergence-rate} compares the Lloyd's algorithm with \gene (i.e., the non-private version of \privGene without considering \difp). For Lloyd, we reuse the initial centroids generated in Section \ref{sec:existingAppExpt}. Given a dataset, we run Lloyd on the 30 sets of initial centroids generated for the dataset, and report the average NICV.
Generally, \gene overtakes Lloyd as the number of iterations increases and finally converges to the global optimum. However, Lloyd improves its performance much faster than \gene in the first few iterations, and converges to the global optimal (or local optimum) more quickly. For example, in the Image dataset, Lloyd reaches the best baseline after three iterations, while the \gene needs more than 10 iterations to achieve the same.
The second reason that \privGene is inferior to \dpl is the low privacy budget allocated to select a parameter (i.e., a set of $k$ cluster centroids) from the candidate set. In each iteration \privGene selects 10 parameters, and the total number of iterations is at least 8. Thus, the privacy budget allocated to select a single parameter is at most $\epsilon/80$. Therefore, \privGene has reasonable performance only for big $\epsilon$ value.


\begin{figure*}[!htb]
	\begin{tabular}{ccc}
	\includegraphics[width = 2.2in]{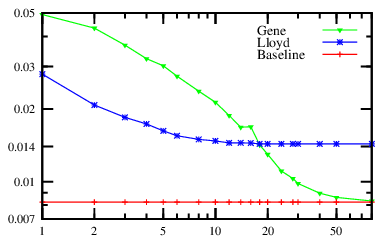} &
	\includegraphics[width = 2.2in]{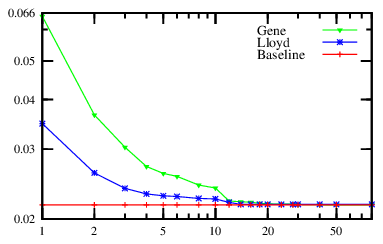} &	
	\includegraphics[width = 2.2in]{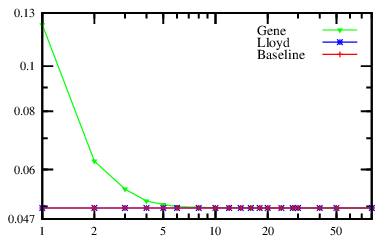}\\
	(a) S1 [$d = 2$, $k = 15$] & (b) Gowalla [$d = 2$, $k = 5$] &  (c) TIGER [$d = 2$, $k = 2$]
	\end{tabular}

	\begin{tabular}{ccc}
	\includegraphics[width = 2.2in]{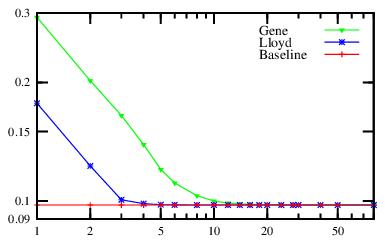} &
	\includegraphics[width = 2.2in]{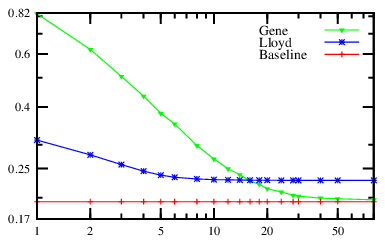} &	
	\includegraphics[width = 2.2in]{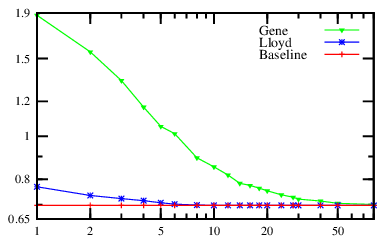}\\
	(d) Image [$d = 3$, $k = 3$] & (e) Adult-num [$d = 6$, $k = 5$] &  (f) Lifesci [$d = 10$, $k = 3$]
	\end{tabular}

	\caption{The comparison of the convergence rate of the genetic algorithm based k-means and Lloyd algorithm. x-axis: number of iterations in log-scale, y-axis: NICV in log-scale.}\label{fig:convergence-rate}
\end{figure*}

\section{The Hybrid Approach}\label{sec:hybrid}

Experimental results in Section~\ref{sec:towardHybrid} establish that \dpl is the best performing interactive method; however, it still under-performs EUGkM.  Recall that EUGkM publishes a private synopsis of the the dataset, and thus enables other analysis to be performed on the dataset, beyond $k$-means.  This means that currently the non-interactive method has a clear advantage over interactive methods, at least for $k$-means clustering.  

An intriguing question is ``Whether EUGkM is the best we can do for $k$-means clustering?'' In particular, can we further improve \dpl?  Recall that there are two key issues that greatly affect the accuracy of \dpl: the number of iterations and the choice of initial centroids.  In fact, these two are closely related.  If the initially chosen centroids are very good and close to the true centroids, one only needs perhaps one more iteration to improve it, and this reduction in the number of iterations would mean little noise is added.  Now if only we have a method to choose really good centroids in a differentially private way, then we can use part (e.g., half) of the privacy budget to get those initial centroids, and the remaining privacy budget to run one iteration of \dpl to further improve it.  

In fact, we do have such a method.  EUGkM does it.  This leads us to propose a hybrid method that combines non-interactive EUGkM with interactive \dpl.  We first use half the privacy budget to run EUGkM, and then use the centroids outputted by EUGkM as the initial centroids for one round of \dpl.  Such a method, however, may not actually outperform EUGkM, especially when the privacy budget $\epsilon$ is small, since then one round of \dpl may actually worsen the centroids.  Therefore, when $\epsilon$ is small, we should stick to the EUGkM method, and only when $\epsilon$ is large enough should we adopt the EUGkM+\dpl approach.  In order to determine what $\epsilon$ is large enough, we analyze how the errors depend on the various parameters in \dpl and in EUGkM.

\subsection{Error Study of \dpl}\label{sec:dpl_Analysis}

DPLloyd adds noises to each iteration of updating centroids.  To study the error behavior of DPLloyd due to the injected Laplace noises, we focus on analyzing the mean squared error (MSE) between noisy centroids and true centroids in one iteration.  

Consider one centroid and its update in one iteration.  The true centroid's $i$'th dimension should be $o_i=\frac{S_i}{C}$, where $C$ is the number of data points in the cluster and $S_i$ is the sum of $i$'th dimension coordinates of data points in the cluster.  Consider the noisy centroid $\widehat{o}$; its $i$'th dimension is $\widehat{o_i}=\frac{S_i+\Delta S_i}{C +\Delta C}$, where $\Delta C$ is the noise added to the count and $\Delta S_i$ is the noise added to the $S_i$.  The MSE is thus:
\begin{equation} \label{eq:DPLMSE}
\mse{\widehat{o}}  = \E\left[ \sum_{i=1}^{d} \left(\frac{S_i + \Delta S_i}{C + \Delta C} - \frac{S_i}{C}\right)^2 \right]
\end{equation}


Derivation based on the above formula gives the following proposition.

\begin{proposition}\label{thm:DPLMSE}
In one round of \dpl, the MSE is
\begin{equation*}\label{eqn:dpl_mse}
\Theta\left(\frac{(kt)^2d^3}{(N\epsilon)^2}\right).
\end{equation*}
\end{proposition}

\begin{proof}
Let us first consider the MSE on the $i$-th dimension.
$$
\begin{array}{rl}
 & \hspace{-0.4cm}\mse{\widehat{o_i}} = \E\left[ \left( \frac{S_i + \Delta S_i}{C + \Delta C} - \frac{S_i}{C}\right)^2 \right]
\\
 \approx &\hspace{-0.3cm} \E\left[ \left(\frac{C \Delta S_i  - S_i \Delta C}{C^2} \right)^2 \right]\nonumber
\\
 = &\hspace{-0.3cm} \frac{\E[(\Delta S_i)^2]}{C^2} +  \frac{\E[S_i^2 (\Delta C)^2]}{C^4} + \frac{2CS_i\E[\Delta S_i \Delta C]}{C^4} 
\\
 = & \frac{\Var{\Delta S_i}}{C^2} +  \frac{S_i^2\Var{\Delta C}}{C^4}
\end{array}
$$
The last step holds, because $\Delta S_i$ and $\Delta C$ are independent zero-mean Laplacian noises and the following formulas hold:
$$
\begin{cases}
\E[\Delta S_i \Delta C]=0\\
\E[(\Delta S_i)^2]=\E[(\Delta S_i)^2]-(\E[\Delta S_i])^2=\Var{\Delta S_i}\\
\E[(\Delta C)^2] =\E[(\Delta C)^2]-(\E[\Delta C])^2= \Var{\Delta C},
\end{cases} \nonumber
$$
where $\Var{\Delta S_i}$ and $\Var{\Delta C}$ are the variances of $\Delta S_i$ and $\Delta C$, respectively. 

Suppose that on average $\frac{|S_i|}{2r\cdot C}=\rho$, where $[-r,r]$ is the range of the $i$'th dimension. That is, $\rho$ is the normalized coordinate of $i$-th dimension of the cluster's centroid. Furthermore, suppose that each cluster is about the same size, i.e., $C\approx \frac{N}{k}$. Then, $\mse{\widehat{o_i}}$ can be approximated as follows:
\begin{eqnarray}
\mse{\widehat{o_i}} & \approx &   \frac{k^2}{N^2} \left( \Var{\Delta S_i} +  (2\beta r)^2 \cdot \Var{\Delta C} \right) \label{eqn:DPLMSEoi}
\end{eqnarray}
\dpl adds to each sum/count function Laplace noise $\Lap{\frac{(dr+1)t}{\epsilon}}$. Therefore, both $\Var{\Delta S_i}$ and $\Var{\Delta C}$ are equal to $\frac{2((dr+1)t)^2}{\epsilon^2}$. From Equation (\ref{eqn:DPLMSEoi}) we obtain
\begin{eqnarray*}
\mse{\widehat{o_i}} & \approx &\frac{k^2}{N^2} \left( \Var{\Delta S_i} +  (2\rho r)^2 \cdot \Var{\Delta C} \right) \\
& = & 2 (1+(2\rho r)^2) \left(\frac{kt(dr+1)}{N\epsilon}\right)^2.
\end{eqnarray*}
As the noise added to each dimension is independent, from Equation \ref{eq:DPLMSE} we know that the MSE is
\begin{eqnarray}\label{eqn:DPLMSE_approx}
\mse{\widehat{o}} = \sum_{i=1}^d \mse{\widehat{o_i}} \approx 2d (1+(2\rho r)^2) \left(\frac{kt(dr+1)}{N\epsilon}\right)^2
\end{eqnarray}
When $r$ is a small constant, this becomes $\Theta\left(\frac{(kt)^2 d^3}{(N\epsilon)^2}\right)$.
\end{proof}

Proposition~\ref{thm:DPLMSE} shows that the distortion to the centroid proportional to $t^2k^2d^3$, while  inversely proportional to $(N\epsilon)^2$. At first glance, this analysis seems to conflict with the experimental result in Figure \ref{fig:synthe-heatmap} (a), where \dpl is much less scalable to $k$ than to $d$. The reason behind is that the performance of \dpl is also affected by the fact that $5$ rounds are not enough for it to converge.  When $k$ increases, converging takes more time, and it is also more likely that choices of initial centroids lead to local optima that are far from global optimum.

\subsection{Error Study of EUGkM}\label{sec:non_inter_analysis}

Non-interactive approach partitions a dataset into a grid of $M$ uniform cells. Then, it releases private synopses for the cells, and runs \km clustering on the synopses to return the cluster centroids. 
%
%
%
Similar to the error analysis for \dpl, we analyze the MSE. Let $o$ be the true centroid of a cluster, and $\widehat{o}$ be its estimator computed by a non-interactive approach. The MSE between $\widehat{o}$ and $o$ is composed of two error sources. First, the count in each cell is inaccurate after adding Laplace noise. This results in the variance (i.e., $\Var{\widehat{o}}$) of $\widehat{o}$ from its expectation $\EE{\widehat{o}}$. Second, we no longer have the precise positions of data points, and only assume that they occur at the center in a cell. Thus, the expectation of $\widehat{o}$ is not equal to $o$, resulting in a bias (i.e., $\BI{\widehat{o}}$). The MSE is the combination of these two errors.
\begin{equation}\label{eqn:mse}
\mse{\widehat{o}}  =  \Var{\widehat{o}} + (\BI{\widehat{o}})^2
\end{equation}

\mypara{Analyzing the variance.} We assume that each cluster has a volume that is $\frac{1}{k}$ of the total volume of the data space, and has the shape of a cube. In $d$-dimensional case, the width of the cube is $w=\frac{2r}{\sqrt[d]{k}}$. Suppose that the {\em geometric} center\footnote{Note that this is not the cluster centroid.} of the cube is $\tau_i$. Let $T$ be the set of cells included in the cluster. For each cell $t \in T$, we use $c_t$ to denote the number of tuples in $t$, $t_i$ to denote the $i$'th dimension coordinate of the center of cell $t$, and $\nu_t$ to denote the noise added to the cell size.  Let $\widehat{o_i}$ be the $i$-th dimension of the noisy centroid. Then, the variance of $\widehat{o_i}$ is
$$
\begin{array}{ll}
& \Var{\widehat{o_i}}  =   \Var{\widehat{o_i}-\tau_i}\vspace{1.5mm}
 \\
= & \Var{\frac{\sum_{t \in T} t_i (c_t + \nu_t)}{\sum_{t \in T} (c_t + \nu_t)}- \tau_i} \vspace{1.5mm}
 \\
= & \Var{\frac{\sum_{t \in T} (t_i - \tau_i) (c_t + \nu_t)}{\sum_{t \in T} (c_t + \nu_t)}}\vspace{1.5mm}
\\
\approx & \frac{1}{C^2}\sum_{t\in T} \left((t_i - \tau_i)^2 \cdot\Var{c_t + \nu_t}\right).
\\
\end{array}
$$
In the above, the first step follows because $\tau_i$ as the cube geometric center is a constant. The last step is derived by assuming $\sum_{t \in T} (c_t + \nu_t)\approx C$, that is, the noisy cluster size is approximately equal to the original cluster size $C$.

We can see that within the cube, different cells' contribution to the variance is not the same.  Basically, the closer a cell is to the cube center, the less its contribution. The contribution is proportional to the squared distance to the cube center. We thus approximate the variance as follows:
\begin{eqnarray*}
\Var{\widehat{o_i}}&\approx &\frac{1}{C^2} \int_{-\frac{w}{2}}^{\frac{w}{2}}x^2 \left(\frac{M}{(2r)^d}w^{d-1}\frac{2}{\epsilon^2}\right) \mathrm{d} x
 \\
&= & \frac{2Mr^2}{3C^2\epsilon^2k^{\frac{d+2}{d}}}.
\end{eqnarray*}
In the above integral, $x$ in the first term is the distance from a cell center to the cube center (i.e., $t_i-\tau_i$). The second term $\frac{M}{(2r)^d}$ is the number of cells per unit volume, and $w^{d-1}$ is the volume of the $(d-1)$-dimensional plane that has a distance of $x$ to the cube center. The last term $\frac{2}{\epsilon^2}$ is the variance of the cell size (i.e., $\Var{c_t + \nu_t}$). Suppose that clusters are of equal size, that is, $C=\frac{N}{k}$. Then, the variance of the noisy centroid by summing all the $d$ dimensions  is
\begin{equation}\label{eqn:nonInter_variance}
\Var{\widehat{o}}\approx\frac{2dMr^2k^{\frac{d-2}{d}}}{3N^2\epsilon^2}
\end{equation}

The analysis shows that the variance of the \eugkm is proportional to $\frac{M}{(N\epsilon)^2}$.  \eugkm sets $M$ to $\left(\frac{N\epsilon}{10}\right )^{\frac{2d}{2+d}}$.  Plugging it into Equation~\ref{eqn:nonInter_variance}, we get that the variance of \eugkm is \emph{inversely proportional} to $\left(N\epsilon\right)^{\frac{4}{2+d}}$.

\mypara{Analyzing the bias}.
Let $x_i$ be the $i$'th dimension coordinate of a tuple $x$. Then, the bias of $\widehat{o_i}$ is
$$
\begin{array}{ll}
& \BI{\widehat{o_i}} = \EE{\widehat{o_i}}-o_i \vspace{1.5mm}
 \\
= & \EE{\frac{\sum_{t \in T} t_i (c_t + \nu_t)}{\sum_{t \in T} (c_t + \nu_t)}} - \frac{\sum_{t \in T} \sum_{x\in t} x_i }{\sum_{t \in T} c_t} \vspace{1.5mm}
 \\
\approx & \frac{\sum_{t \in T} \sum_{x\in t} (t_i-x_i) }{C},
\\
\end{array}
$$
where the last step is developed by approximating $\sum_{t \in T} (c_t + \nu_t)$ to the cluster size $C$.

The bias developed in the above formula is dependent on data distribution. Its precise estimation requires to access real data. We thus only estimate its upper bound. Let $q_i=t_i-x_i$. Non-interactive approach partitions each dimension into $\sqrt[d]{M}$ intervals of equal length. Hence, $q_i$ falls in the range of $[-\frac{r}{\sqrt[d]{M}}, \frac{r}{\sqrt[d]{M}}]$, and the upper bound of $\BI{\widehat{o_i}} $ is $\frac{r}{\sqrt[d]{M}}$. Summing all the $d$ dimensions, we obtain the upper bound of squared bias of noisy centroid
\begin{equation}\label{eqn:nonInter_bias2}
(\BI{\widehat{o}})^2\leq\frac{dr^2}{M^{\frac{2}{d}}}.
\end{equation}

The estimation shows that the upper bound of squared bias decreases as a function of $M^{\frac{2}{d}}$. This is consistent with the expectation. As $M$ increases, the data space is partitioned into finer-grained cells. Therefore, the distance between a tuple in a cell to the cell center decreases on average.

\mypara{Comparing \dpl and \eugkm.} We now analyze the performance of \dpl and \eugkm in Figure \ref{fig:non-interactive_VS_interactive}. Equation \ref{eqn:DPLMSE_approx} shows that the MSE of \dpl is inversely proportional to $(N\epsilon)^2$. The MSE of \eugkm consists of variance and squared bias. Plugging $M=\left(\frac{N\epsilon}{10}\right )^{\frac{2d}{2+d}}$ into Equation \ref{eqn:nonInter_variance} and Inequality \ref{eqn:nonInter_bias2},
it follows that the MSE of \eugkm is inversely proportional to $(N\epsilon)^{\frac{4}{2+d}}$. This explains why the NICV of \dpl, which is inversely proportional to $(N\epsilon)^2$ drops much faster than that of \eugkm as $\epsilon$ grows. It also explains why \dpl has better performance on `big' dataset (e.g., the TIGER dataset).

The MSE of \eugkm is inversely proportional to $(N\epsilon)^{\frac{4}{2+d}}$. Thus, it increases exponentially as a function of $d$. Instead, from Equation \ref{eqn:DPLMSE_approx}, it follows that the MSE of \dpl has only cubic growth with respect to $d$. Therefore, in Figure \ref{fig:non-interactive_VS_interactive}, as the dimensionality of dataset increases, \dpl outperforms \eugkm. This also explains in Figure \ref{fig:synthe-heatmap}  why \dpl is more scalable to $d$ than \eugkm.

\subsection{The Hybrid Approach}\label{ssec:hybrid}

Our hybrid approach combines \eugkm and \dpl. Given a dataset and privacy budget $\epsilon$, the hybrid approach first \mbox{checks} whether it overtakes the \dpl method and also the \eugkm method. If this is not the case, the hybrid approach simply falls back to \eugkm. Otherwise, the hybrid approach allocates half privacy budget to \eugkm to output a synopsis and find $k$ intermediary centroids that work well for the synopsis. Then, it runs \dpl for one iteration using the remaining half privacy budget to refine these $k$ centroids.

We use MSE to heuristically determine the conditions, on which the hybrid approach overtakes the \dpl method and also the \eugkm method. Basically, we require that the MSE of the hybrid approach be smaller than those of the other two approaches, since smaller MSE implies smaller error to the cluster centroid.
From Equation \ref{eqn:DPLMSE_approx}, it follows that the MSE of \dpl with full privacy budget is
\begin{equation}\label{formula:DPLopt_MSE_tround}
2d (1+(2\rho r)^2) \left(\frac{kt(dr+1)}{N \epsilon}\right)^2.
\end{equation}
A precise estimation of the MSE of the \eugkm method requires to access the dataset, since the bias depends on the real data distribution. However, we have the approximate variance (Equation \ref{eqn:nonInter_variance}) by setting $M=\left(\frac{N\epsilon}{10}\right )^{\frac{2d}{2+d}}$.
\begin{equation}\label{formula:eugkm_var}
\frac{2dr^2(k)^{\frac{d-2}{d}}}{3\times (10)^{\frac{2d}{2+d}}(N\epsilon)^{\frac{4}{2+d}}}
\end{equation}
One-iteration \dpl with half privacy budget outputs the final $k$ cluster centroids, if it is applied in the hybrid approach. Therefore, we approximate the MSE of the hybrid approach by that of the one-iteration \dpl
\begin{equation}\label{formula:dplmse_1round}
8d (1+(2\rho r)^2) \left(\frac{k(dr+1)}{N \epsilon}\right)^2,
\end{equation}
which is developed by setting $t=1$ and privacy budget to $0.5\epsilon$ in Equation \ref{eqn:DPLMSE_approx}.

Comparing Formulas \ref{formula:DPLopt_MSE_tround} and \ref{formula:dplmse_1round}, it follows that the MSE of the hybrid approach is lower than or equal to that of the \dpl if
\begin{equation}\label{eqn:alpha}
t\geq 2.
\end{equation}

Variance is the lower bound of MSE. Thus, if the MSE of the hybrid approach is equal to or smaller than the variance of the \eugkm method, then it is sure that the hybrid approach has lower MSE. Setting Formula \ref{formula:dplmse_1round} smaller than or equal to Formula \ref{formula:eugkm_var} yields
\begin{equation}\label{eqn:boundaryEpsilon}
\epsilon \geq \left(\frac{X}{Y}\right )^{\frac{2+d}{2d}},
\end{equation}
where
$$
X= 8d (1+(2\rho r)^2) \left(\frac{k(dr+1)}{N}\right)^2,
$$
and
$$
Y = \frac{2dr^2(k)^{\frac{d-2}{d}}}{3\times (10)^{\frac{2d}{2+d}}N^{\frac{4}{2+d}}}.
$$

Inequalities \ref{eqn:alpha} and \ref{eqn:boundaryEpsilon} give the conditions of applying the hybrid approach. Inequality \ref{eqn:alpha} is automatically satisfied since \dpl runs for $t=5$ iterations.

\subsection{Experimental results}

We now compare the hybrid approach with \eugkm and \dpl. The configuration for \eugkm and \dpl is the same as in Section \ref{sec:existingAppExpt}.  For the hybrid approach, we run \eugkm 10 times to output 10 sets of intermediate centroids.  Then we run \dpl 10 times on each intermediate result.  We finally report the average of 100 NICV values. Figure \ref{fig:hybrid} gives the results on the six external datasets. In low dimensional datasets (S1, Gowalla, TIGER, and Image), the hybrid approach simply falls back to \eugkm for small $\epsilon$ value. When $\epsilon$ increases, both the hybrid approach and \eugkm converge to the baseline with the former having slightly better performance. For example, in the Gowalla dataset for $\epsilon=0.7$, the average NICV of the hybrid approach is $0.02172$ and that of \eugkm is $0.02174$.

In higher dimensional datasets (Adult-num and Lifesci), the hybrid approach outperforms the other two approaches in most cases. It is worse than \dpl only for a few small $\epsilon$ values, on which it falls back to \eugkm. There are two possible reasons. The first is that the MSE analysis assumes that datasets are well clustered and each cluster has equal size, but the real datasets are skewed. For example, the baseline approach partitions the Adult-num dataset into 5 clusters, in which the biggest cluster contains 13,894 tuples and the smallest contains 3,160 tuples. The second is that we use the variance of \eugkm as the lower bound of its MSE. Thus, it is possible that the MSE of the hybrid approach (approximated by the MSE of one-iteration \dpl with half privacy budget) is larger than the variance of \eugkm, but actually smaller than its MSE. In such cases, the hybrid approach gives lower NICV if it does not fall back to \eugkm. For example, on the Adult-num dataset for $\epsilon=0.05$, the hybrid approach of falling back to \eugkm has the NICV of $0.370$, while its NICV is $0.244$, if it applies \eugkm plus one-iteration of \dpl.

\begin{figure*}[!htb]
	\begin{tabular}{ccc}
	\includegraphics[width = 2.2in]{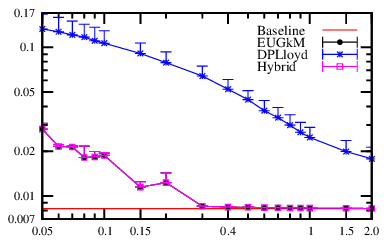} &
	\includegraphics[width = 2.2in]{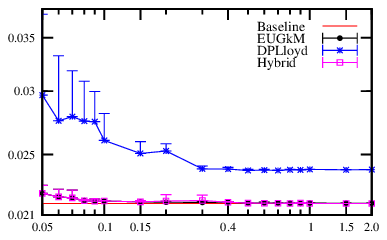} &
	\includegraphics[width = 2.2in]{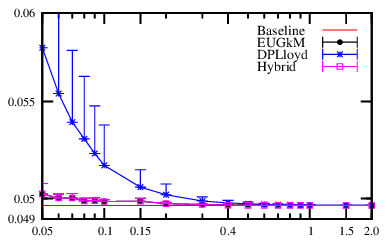}\\
	(a) S1 [$d = 2$, $k = 15$]  & (b) Gowalla [$d = 2$, $k = 5$] & (c) TIGER [$d = 2$, $k = 2$]
	\end{tabular}
	
	\begin{tabular}{ccc}
	\includegraphics[width = 2.2in]{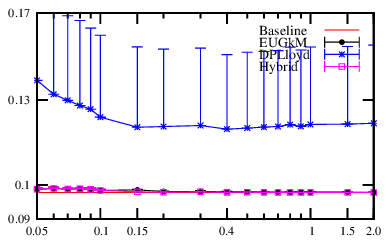}&
	\includegraphics[width = 2.2in]{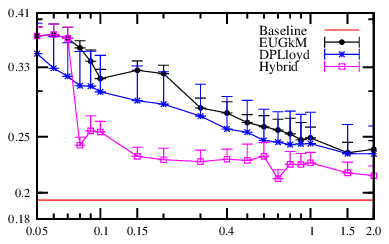}&
	\includegraphics[width = 2.2in]{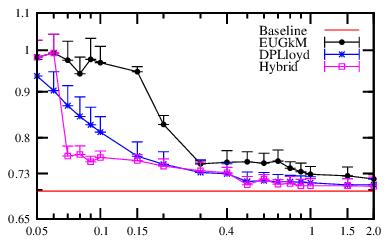}\\
	(d) Image [$d = 3$, $k = 3$] & (e) Adult-num [$d = 6$, $k = 5$] & (f) Lifesci [$d = 10$, $k = 3$]
	\end{tabular}
	
    \caption{The comparison of the Hybrid approach with EUGkM and \dpl. x-axis: privacy budget $\epsilon$ in log-scale. y-axis: NICV in log-scale.}\label{fig:hybrid}
\end{figure*}

We also evaluate the approaches using the synthetic datasets as generated in Section \ref{sec:existingAppExpt}. Figure~\ref{fig:heatmap-eugkm-hybrid} clearly shows that the hybrid approach is more scalable than \eugkm with respect to both $k$ and $d$. This confirms the effectiveness of the hybrid approach.

Figure~\ref{fig:running_time_dplloyd_eugkm} presents the runtime of \dpl and \eugkm on the six external datasets.  We follow the same experiment configuration as in Section~\ref{sec:existingAppExpt}.  As expected, the runtime of \dpl is much lower than that of \eugkm.  This is because \eugkm has to run \km clustering over 30 sets of initial centroids and output the centroids with the best NICV relative to the noisy synopsis.  Another reason is that \dpl sets the number of iterations to 5 while \eugkm runs \km clustering until converge.

\begin{figure}[h]
\begin{tabular}{cc}
\includegraphics[width = 1.8in]{eugkm_nicv.eps}&
\hspace{-0.8cm}\includegraphics[width = 1.8in]{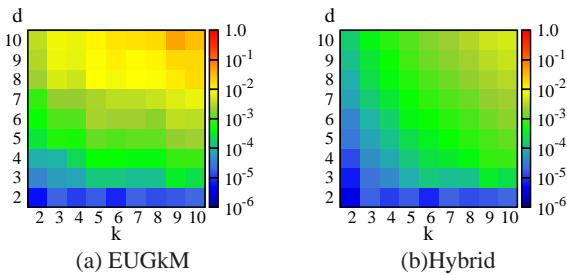}\\
(a) \eugkm & \hspace{-0.8cm} (b)Hybrid
\end{tabular}
\caption{Comparing hybrid and \eugkm by the heatmap}\label{fig:heatmap-eugkm-hybrid}
\end{figure}

\begin{figure}[h]
\includegraphics[width = 3.1in]{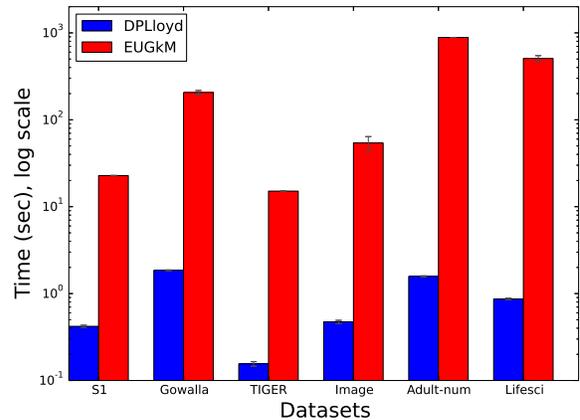}
\caption{Comparing running time between \dpl and \eugkm, $\epsilon = 0.1$}\label{fig:running_time_dplloyd_eugkm}
\end{figure}

\section{Conclusion and Discussions}\label{sec:conclusions}
We have improved the state of art on differentially private $k$-means clustering in several ways. We have introduced non-interactive methods for differentially private $k$-means clustering, and have extensively evaluated and analyzed three interactive methods and one non-interactive methods.  Our proposed EUGkM outperforms existing methods.  We have also introduced the novel concept of hybrid approach to differentially private data analysis, which is so far the best approach to $k$-means clustering.  

Concerning the question of non-interactive versus interactive, the insights obtained from $k$-means clustering are as follows.  
The non-interactive EUGkM has clear advantage, especially when the privacy budget $\epsilon$ is small.  Considering the further advantage that non-interactive methods enable other analysis on the dataset, we would tentatively conclude that non-interactive is the winner in this comparison.  We conjecture that this tradeoff will hold for many other data analysis tasks.  We plan to investigate whether this holds in other analysis tasks.
Also, if one's goal is to improve the accuracy of one $k$-means clustering task as much as possible, then hybrid approaches may be the most promising solution.


\bibliographystyle{abbrv}
\bibliography{privacy}

\end{document}